\documentclass[a4paper,12pt]{article}

\makeatletter
 
  \@addtoreset{equation}{section}
 \makeatother
\usepackage{amsfonts}
\usepackage{amssymb}
\usepackage{mathrsfs}
\usepackage{amsmath,amsthm}
\usepackage{amsmath}

\usepackage{graphicx}
\usepackage{color}
\usepackage{indentfirst}

\usepackage{amssymb}

\begin{document}

\title{An Optimal Execution Problem with a Geometric Ornstein--Uhlenbeck Price Process} 

\author{Takashi Kato
\footnote{Division of Mathematical Science for Social Systems, 
              Graduate School of Engineering Science, 
              Osaka University, 
              1-3, Machikaneyama-cho, Toyonaka, Osaka 560-8531, Japan, 
E-mail: \texttt{kato@sigmath.es.osaka-u.ac.jp}}
}

\date{First Version: July 8, 2011\\ 
This Version: July 29, 2014
}

\newtheorem{theorem}{Theorem}
\newtheorem{lemma}{Lemma}
\newtheorem{proposition}{Proposition}
\newtheorem{step}{Step}
\newtheorem{corollary}{Corollary}
\theoremstyle{definition}
\newtheorem{defn}{Definition}
\newtheorem{rem}{Remark}

\newcommand{\Cov}{\mathop {\rm Cov}}
\newcommand{\Var}{\mathop {\rm Var}}
\newcommand{\E}{\mathop {\rm E}}
\newcommand{\const }{\mathop {\rm const }}
\everymath {\displaystyle}

\newcommand{\ruby}[2]{
\leavevmode
\setbox0=\hbox{#1}
\setbox1=\hbox{\tiny #2}
\ifdim\wd0>\wd1 \dimen0=\wd0 \else \dimen0=\wd1 \fi
\hbox{
\kanjiskip=0pt plus 2fil
\xkanjiskip=0pt plus 2fil
\vbox{
\hbox to \dimen0{
\small \hfil#2\hfil}
\nointerlineskip
\hbox to \dimen0{\mathstrut\hfil#1\hfil}}}}

\def\qedsymbol{$\blacksquare$}
\renewcommand{\thefootnote }{\fnsymbol{footnote}}
\renewcommand{\refname }{References}

\everymath {\displaystyle}

\maketitle 

\begin{abstract}
\footnote[0]{Mathematical Subject Classification (2010) \  91G80, 93E20, 49L20}\\
\footnote[0]{JEL Classification (2010) \ G33, G11}
We study an optimal execution problem in the presence of market impact 
where the security price follows a geometric Ornstein--Uhlenbeck process, which implies the mean-reverting property, 
and show that the optimal strategy is a mixture of 
initial/terminal block liquidation and gradual intermediate liquidation. 
The mean-reverting property describes a price recovery effect 
that is strongly related to the resilience of market impact, 
as described in several papers that have studied optimal execution in a limit order book (LOB) model. 
It is interesting that despite the fact that the model in this paper is different from the LOB model, 
the form of our optimal strategy is quite similar to those obtained for an LOB model. 
Moreover, we discuss what properties cause gradual liquidation as an optimal strategy 
by studying various cases 
and find out that not only ``convexity of market impact function'' but also ``price recovery effect'' 
(or, in other words, transience of market impact) are essential to make a trader 
execute the security gradually to mitigate the effect of market impact.  
\end{abstract}

\section{Introduction}\label{intro} 

The basic framework of the optimal execution (liquidation) problem was established in 
Bertsimas and Lo~\cite{Bertsimas-Lo}, 
and the theory of optimal execution has been developed by Almgren and Chriss~\cite{Almgren-Chriss}, 
He and Mamaysky~\cite{He-Mamaysky}, Huberman and Stanzl~\cite{Huberman-Stanzl2}, 
Subramanian and Jarrow~\cite{Subramanian-Jarrow}, and many others 
(see also Gatheral and Schied~\cite {Gatheral-Schied}, in which we survey dynamic models of execution problems). 
Optimal execution problems arise naturally in trading operations, 
such as when a trader tries to execute a large trade for a security. 
In these cases, the trader should take care about liquidity problems 
and, in particular, should not neglect market impact (MI), which plays an important role in execution cost. 
Here, MI means the effect that a trader's investment behavior has on security prices. 

To study the MI of a trader's execution policy, 
we consider a case where a trader sells held shares of the security 
after predicting a decrease in the price of the security. 
In a frictionless market, a risk-neutral trader should sell all shares as soon as possible, 
and so the optimal strategy is block liquidation at the initial time. 
However, in real markets, traders tend to liquidate a position over time. 
The factors that lead to gradual liquidation are therefore important. 

Convexity of MI is one factor that would dissuade traders from block liquidation. 
As shown in examples in Kato~\cite{Kato}, 
a risk neutral trader in a market of the Black--Scholes type 
will gradually liquidate if MI follows a quadratic function; in contrast, block liquidation is optimal when MI is linear. 
However, many traders in the real market execute their sales over time despite recognizing that 
MI is not always convex. 

Risk aversion also will affect a trader's execution policy, 
providing an incentive to trade over a longer period. 
Schied and Sch\"oneborn~\cite{Schied-Schoeneborn} consider the optimal strategies when the utility function rewards risk aversion
and clarify the relation between the degree of risk aversion and the form of the optimal strategy. 
Additionally, He and Mamaysky~\cite{He-Mamaysky} treat execution problems in a Black--Scholes-type model 
with a linear MI function and numerically derive some no-trading regions of optimal strategies; in such regions the optimal strategies are not all block liquidation. 
Howoever, from Lions and Lasry~\cite{Lions-Lasry} and Kato~\cite{Kato}, 
we know that the optimal strategy under a linear MI function is not 
gradual in several cases: 
it is block liquidation at the initial time even when the trader is risk averse 
so long as the risk-adjusted drift coefficient of the security price is nonpositive. 

Another important motive for liquidation over time is that, due to the effect of MI, 
a security price may recover after a downward movement in price. 
Such a phenomenon is called a ``price recovery effect,'' 
and such effects implicitly describe transient MI 
(see Gatheral and Schied~\cite{Gatheral-Schied} for details). 
In this paper, to consider a price recovery effect, we focus on the case where the process of a security price has 
the mean-reverting property, and, in particular, we focus on when it follows a geometric Ornstein--Uhlenbeck (OU) process. 
We explicitly solve the optimization problem with static execution strategy 
for a linear MI and show that 
the optimal strategy is a mixture of initial/terminal block liquidation 
and gradual intermediate liquidation. 
Our study in this paper is also shown to be a representative case 
of when gradual liquidation is necessary in the framework of Kato~\cite {Kato} 
even with a linear MI and risk-neutral trader. 

Our results are related to those of studies of execution problems in limit order book (LOB) models. 
In a LOB model, sales by a trader decrease buy limit orders, thereby
temporarily expanding the bid--ask spread, and new buy limit orders appear over time, causing 
the bid--ask spread to shrink as time passes. 
The problem of minimizing expected execution cost in a block-shaped LOB model 
with exponential resilience of MI is studied in Obizhaeva and Wang~\cite{Obizhaeva-Wang}. 
A mathematical generalization of the results of Obizhaeva and Wang~\cite{Obizhaeva-Wang} is given in 
Alfonsi et al.~\cite{Alfonsi-Fruth-Schied} and Predoiu et al.~\cite {Predoiu-et-al}. 
Additionally, Makimoto and Sugihara~\cite{Makimoto-Sugihara} treat a model of optimal execution 
under stochastic liquidity. 
It is interesting that despite the fact that the model in this paper is different from the LOB model, 
the form of optimal execution strategies in our model become quite similar 
to the optimal strategies found in papers focusing on an LOB model. 
Indeed, when the security price process has no volatility, 
the form of our optimal strategy coincides with those in 
Alfonsi et al.~\cite{Alfonsi-Fruth-Schied} and Obizhaeva and Wang~\cite{Obizhaeva-Wang}: 
the rate of intermediate liquidation is constant. 
When the volatility is larger than zero, 
the rate decreases over time, as is found in Makimoto and Sugihara~\cite{Makimoto-Sugihara}. 

This paper is organized as follows. 
In Section \ref {sec_model}, we introduce our model settings. 
In Section \ref{sec_gOU}, we explicitly solve our optimization problem and 
give the forms of optimal strategies. 
We additionally discuss essential properties of MI that induce gradual liquidation. 
Section \ref{sec_conclusion} summarizes our study. 
Section \ref {sec_appendix} is an appendix, 
where the 
derivation of our model from discrete-time models (Section \ref{sec_derivation}) and 
the proofs of our results (Section \ref{sec_proofs}) are given.

\section{The Model}\label{sec_model}

Our model is based on that of Kato~\cite{Kato}. 
Let $(\Omega ,\mathcal {F}, \allowbreak (\mathcal {F}_t)_{t\geq 0}, P)$ 
be a filtered space satisfying the usual conditions (i.e., 
$(\mathcal {F}_t)_t$ is right continuous and $\mathcal {F}_0$ contains 
all $P$-null sets), and let 
$(B_t)_{t\geq 0}$ be a standard one-dimensional $(\mathcal {F}_t)_t$-Brownian motion. 
We consider a simple market model in which there are only two financial assets: cash and a security. 
We assume that the risk-free rate of return is zero, so the price of cash is always $1$. 
We study the execution problem of a single trader who has $\Phi _0\geq 0$ shares of the security. 

First, we prepare the class of admissible execution strategies. 
Let $T > 0$ be a time horizon. 
We assume without loss of generality that $T = 1$. 
We denote by 
$\mathcal {A}_T(\Phi _0)$ the set of 
Borel-measureable functions 
$(\zeta _r)_{0\leq r\leq T}$\vspace{2mm} such that 
\begin{itemize}
 \item [(a.)] $\zeta _r\geq 0$ for each $r\in [0, T]$,\ and 
 \item [(b.)] $\int ^T_0\zeta _rdr\leq \Phi _0$. 
\end{itemize}
Here, $\zeta _r$ is regarded as the execution speed at time $r$: 
hence, at time $r$, the instantaneous sales volume is $\zeta _rdr$. 
Condition (a.) means that the trader executes only sell orders. 
Moreover, by (b.), the trader cannot sell more than $\varphi$ shares, 
and so short selling is prohibited in our model 
(see also Section 2 in \cite {Kato}). 

Now, we define a value function that corresponds to the trader's optimization problem. 
In this paper, we treat mainly optimization of the expected cost; 
that is, the trader tries to maximize expected proceeds. 
For $t\in [0,1], (w,\varphi ,s)\in D:= \mathbb {R}\times [0, \Phi _0]\times [0, \infty )$, 
we define 
\begin{eqnarray}\label{value_RN}
V_t(w,\varphi ,s ; u_{\mathrm {RN}}) = 
\sup _{(\zeta _r)_{r}\in \mathcal {A}_t(\varphi)}
\E [W_t], 
\end{eqnarray}
subject to 
\begin{eqnarray}\label{fluc_W}
dW_r &=& \zeta_rS_rdr, \\\label{fluc_X}
dX_r &=& \sigma dB_r+ \beta (F - X_r)dr - \alpha \zeta_rdr, \\\label{exp_X}
S_r &=& \exp(X_r)
\end{eqnarray}
and 
\begin{eqnarray}\label{init_cond}
(W_0, X_0) = (w, \log s). 
\end{eqnarray}
(When $s = 0$, we have $S_r\equiv 0$ and $X_r \equiv  -\infty$.) 
Note that the set $\mathcal {A}_t(\varphi )$ of admissible strategies 
is defined in the same manner as $\mathcal {A}_T(\Phi _0)$ is defined. 
Here, $W_r$, $S_r$, and $X_r$ denote, respectively, 
the trader's cash holdings, 
the security's price at time $r$, and the security's log-price at time $r$. 
A risk-neutral trader is assumed to have utility function $u_{\mathrm {RN}}$. 
The parameter $\alpha \geq 0$ characterizes a linear permanent MI: 
when the trader sells $\zeta _rdr$ at time $r$, the log-price decreases by $\alpha \zeta _rdr$. 
When $\alpha = 0$, there is no MI; 
$(X_r)_r$ follows the OU process with mean-reversion speed $\beta > 0$ and 
volatility $\sigma \geq 0$. 
We remind the reader that the transient MI is described by the mean-reverting property of $(X_r)_r$. 
Here, we can write down the explicit form of the solution of (\ref {fluc_X}):  
\begin{eqnarray}\label{fluc_X_general}
X_r = e^{-\beta r}\log s + (1-e^{-\beta r})F - \alpha e^{-\beta r}\int ^r_0e^{\beta v}\zeta _vdv + 
\sigma e^{-\beta r}\int ^r_0e^{\beta v}dB_v. 
\end{eqnarray}

\begin{rem} 
{\rm As in \cite {Kato}, 
we restrict neither the security price process nor the MI function 
to specific forms such as the OU process and the linear MI, respectively. 
In general, the process $(X_r)_r$ is given as the solution of the following stochastic differential equation: 
\begin{eqnarray}\label{fulc_X_general}
dX_r &=& \sigma (X_r)dB_r+ b(X_r)dr - g(\zeta  _r)dr, 
\end{eqnarray}
where $\sigma , b : \mathbb {R} \longrightarrow \mathbb {R}$ are 
Lipschitz-continuous functions and 
$g : [0, \infty )\longrightarrow [0, \infty )$ is a convex function. 
Moreover, the set of admissible strategies is generalized to the set of 
adaptive strategies (see Section \ref {sec_derivation}). 
Our model is just an example of a model from \cite{Kato}, 
but there are some technical gaps: 
in \cite{Kato}, $\sigma$ and $b$ are assumed to be bounded, 
whereas, 
$b(x) = \beta (F - x)$ in (\ref {fluc_X}) is not bounded. 
However, results similar to those in \cite{Kato} will still be obtained; 
these are derived in Section \ref {sec_derivation}. }
\end{rem}

\begin{rem}
{\rm When considering optimal execution problems for a risk-neutral trader, 
that is, when the trader's purpose is to maximize the expected proceeds of execution 
(equivalently, to minimize the expected execution cost), 
adaptive (stochastic) optimal strategies often become static (deterministic) strategies 
(see Alfonsi et al.~\cite{Alfonsi-Fruth-Schied}, 
Kato~\cite{Kato}, Kuno and Ohnishi \cite{Kuno-Onishi}, 
and Schied and Zhang~\cite{Schied-Zhang}, for instance). 
Moreover, there are several papers that study the optimal execution problem 
with static strategies 
(Almgren and Chriss~\cite{Almgren-Chriss}, Bertsimas and Lo~\cite{Bertsimas-Lo}, 
Konishi and Makimoto~\cite{Konishi-Makimoto}, and 
Makimoto and Sugihara \cite{Makimoto-Sugihara}, etc.). 
In these cases, the trader predetermine a strategy 
before starting execution at $t = 0$. 
Such a static optimal strategy is often called an 
Implementation Shortfall (IS) strategy. 
Following the above circumstances, we treat mainly optimal execution problems with static strategies in this paper, 
except for in Section \ref {sec_derivation}. }
\end{rem}

\section{Main Results}\label{sec_gOU}

In this section, we give explicit forms for the value function and the optimal strategies 
when the security holdings $\varphi $ is small enough or large enough. 
Further, we discuss which properties will cause gradual execution to be optimal under MI. 
For brevity, we set $y = \sigma ^2/(4\beta )$ and $z = \log s - F$. 
We assume $z > 2y (\geq 0)$ so that the security price falls to the fundamental value $e^F$ as time passes.

\subsection{No MI case}

First, we introduce the forms of optimal strategies 
when there is no MI (i.e., when $\alpha = 0$). 

\begin{theorem}\label{th_no_MI}
If $\alpha = 0$, then $V_t(w, \varphi , s ; u_\mathrm {RN}) = w + \varphi s$. 
\end{theorem}

In this case, the trader's (nearly) optimal strategy is given by 
\begin{eqnarray}\label{def_hat_I}
\hat{\zeta }^{0,\delta }_r = \frac{\varphi }{\delta }1_{[0, \delta]}(r) 
\end{eqnarray}
with $\delta \rightarrow 0$. 
More precisely, if we denote by $(\hat{W}^{\delta }_r)_r$ the 
corresponding process of cash holdings, it follows that 
\begin{eqnarray}\label{conv_w}
\E [\hat{W}^{\delta }_t] \longrightarrow 
V_t(w, \varphi , s ; u_\mathrm {RN}), \ \ \delta \downarrow 0. 
\end{eqnarray}
We call such a strategy an ``almost block liquidation'' at the initial time (see also Remark 5.3 in \cite {Kato}). 

\begin{rem}
{\rm We can solve the optimization problem even when $z \leq 2y$. 
Indeed, if $(0 < ) 2e^{-\beta t}y < z \leq 2y$, then
an optimal strategy is given by 
$\hat{\zeta }^{t^*,\delta }_r = \frac{\varphi }{\delta }1_{[t^*, t^* + \delta ]}(r)$, 
which is the almost block liquidation at time $t^*$, 
where $t^*=\log (2y/z)/\beta$. 
Moreover, if $z \leq 2e^{-\beta t}y$, 
then the optimal strategy is the terminal almost block liquidation 
$\hat{\zeta }^{t,\delta }_r = \frac{\varphi }{\delta }1_{[t - \delta , t]}(r)$. 
Therefore, in each case when the market is fully liquid, 
the optimal strategy is block liquidation. }
\end{rem}

\subsection{The case of small $\varphi$}

For the remainder of the paper, we assume that the MI function is non-trivial and linear; that is, we assume that $\alpha > 0$. 
In this subsection, we study cases where $\varphi$ is small. 
In the previous subsection, we see that 
a trader in a fully liquid market (i.e., $\alpha = 0$) 
should sell all securities at the initial time. 
In fact, when $\varphi$ is small enough, the trader's optimal policy is almost the same as 
initial block liquidation. 

\begin{theorem} \ \label{th_phi_small}
If $\varphi \leq  (z - 2y)/\alpha $, then 
\begin{eqnarray}\label{value_fnc_initial}
V_t(w, \varphi , s ; u_\mathrm {RN}) = w + \frac{1 - e^{-\alpha \varphi}}{\alpha}s. 
\end{eqnarray}
\end{theorem}

The form of (\ref {value_fnc_initial}) is the same as that in Theorem 8 of Kato~\cite{Kato}. 
The trader's (nearly) optimal strategy is also given by (\ref {def_hat_I}). 
This result implies that if the trader has a small number of shares of the security, 
then MI provides no incentive to take liquidate gradually.

\subsection{The case of large $\varphi$}

When $\varphi $ is not small, the assertion of Theorem \ref{th_phi_small} fails. 
The trader's selling accelerates the speed of decrease of the security price, 
and so a quick liquidation may be non-optimal because of  the effect of MI. 
Moreover, if the trader's execution makes the price drop below $e^F$, 
then the price will recover to $e^F$ by delaying the sale. 
This gives the trader an incentive to liquidate gradually. 
Our purpose in this subsection is to derive an explicit (nearly) optimal execution strategy. 

\subsubsection{No volatility case}\label{sec_const}

Firstly, we study the special case where $\sigma  = 0$. 
This setting gives a price model that is unrealistic and not meaningful in an actual market 
because the random fluctuation of the security price is ignored. 
However, doing so yields an interesting similarity to 
the results of other studies of optimal execution. 

Let $P(x) = e^{-\alpha x}(1 - \alpha x)$. 
Since the function $P$ is strictly decreasing on $(-\infty , 2 / \alpha ]$, 
we can define its inverse function as $P^{-1} : [-e^{-2}, \infty ) \longrightarrow (-\infty , 2 / \alpha]$. 

We assume that the security holdings $\varphi $ is larger than $z / \alpha$. 
We define the function $C(p) = C_{t, \varphi }(p)$, $p\in \mathbb {R}$ as 
\begin{eqnarray*}
C(p) = 
\exp (\alpha (t\beta + 1)p - \alpha \varphi - t\beta z) + \alpha p - z - 1.  
\end{eqnarray*}
Since $C(p)$ is strictly increasing and $C(z/\alpha ) < 0 < C\left( (\varphi -z/\alpha ) / (1 + \beta t)\right)$, 
the equation $C(p) = 0$ has a unique solution $p^* = p^*(t, \varphi ) \in (\varphi - z/\alpha, (\varphi -z/\alpha ) / (1 + \beta t))$. 
Then, we can show that the following theorem is true.

\begin{theorem} \ \label{eg_OU_const}
If $\varphi > z/\alpha  ( > 0 )$, 
then it holds that 
\begin{eqnarray}\label{th_value_fnc_const}
V_t(w, \varphi ,s ; u_{\mathrm {RN}}) = w + 
\frac{1 - e^{-\alpha (p^* + q^*)}}{\alpha }s + tse^{-\alpha p^*}\zeta^*, 
\end{eqnarray}
where $\zeta ^* = \zeta ^*(t, \varphi )$ and $q^* = q^*(t, \varphi )$ are given by 
\begin{eqnarray*}
\zeta ^* = \beta (p^* - z/\alpha), \ \ 
q^* = \varphi - p^* - t\zeta^*. 
\end{eqnarray*}
\end{theorem}

Note that one can easily check that $p^*, \zeta ^*, q^* > 0$. 
We can then construct a nearly optimal strategy as follows (with $\delta \downarrow 0$): 
\begin{eqnarray}\label{nearly_optimal_strategy0}
\hat{\zeta }^\delta _r = \frac{p^*}{\delta }1_{[0,\delta ]}(r) + 
t\zeta ^* + \frac{q^*}{\delta }1_{[t - \delta , t]}(r). 
\end{eqnarray}
That is, 
(\ref {conv_w}) holds for 
$(\hat{W}^\delta _r)_r$ given by (\ref {fluc_W}) with $(\hat {\zeta }^\delta _r)_r$. 

The strategy $(\hat{\zeta }^\delta _r)_r$ consists of three terms. 
The first term in the right-hand side of (\ref {nearly_optimal_strategy0}) 
corresponds to initial (almost) block liquidation. 
The trader should sell $p^*$ shares of the security at the initial time by selling infinitesimal pieces at infinitesimal intervals to avoid a decrease in the execution proceeds. 
The second term corresponds to gradual liquidation. 
The trader executes the selling gradually until the time horizon at constant speed $\zeta^*$. 
Finally, the trader completes liquidation by selling the remaining shares by 
terminal (almost) block liquidation, which corresponds to the third term. 
So, the nearly optimal strategy is a mixture of both block liquidation and gradual liquidation. 
We point out, in particular, that gradual liquidation is necessary for optimality in this case. 
Figure \ref {fig_OU_const} shows an image representing the form of the optimal strategy. 
In fact, the security price on the interval $(\delta , 1 - \delta )$ is equal to 
a constant $se^{-\alpha p^*}$.

\begin{figure}[t]
\begin{center}
\includegraphics[height = 4.69cm,width=7.296cm]{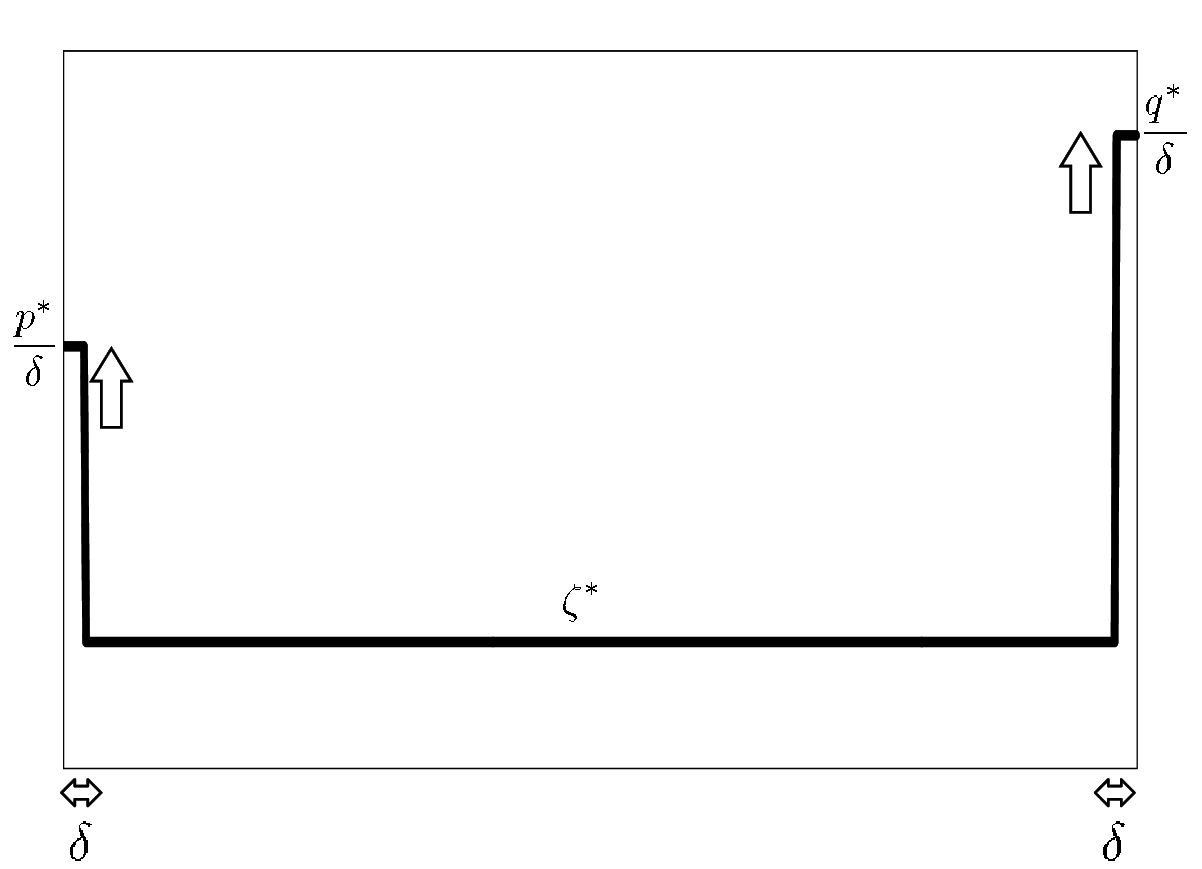}
\includegraphics[height = 4.5cm,width=7cm]{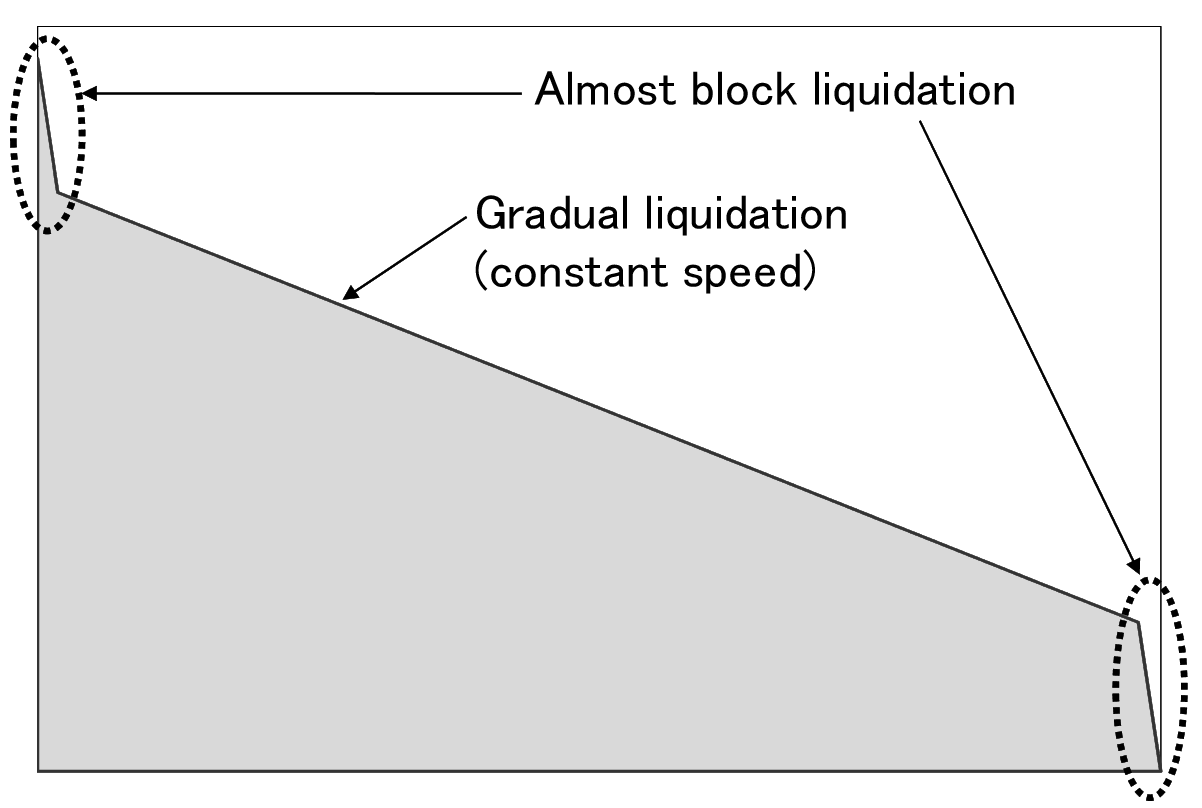}
\caption{Shape of a nearly optimal strategy $(\zeta ^\delta _r)_r$ (left) and 
the corresponding amount of security holdings (right), both when $\sigma = 0$. 
The horizontal axes are time ($r$). }
\label{fig_OU_const}
\end{center}
\end{figure}

This result is quite similar to that of Alfonsi et al.~\cite{Alfonsi-Fruth-Schied} and Obizhaeva and Wang~\cite{Obizhaeva-Wang}, 
despite the fact that there are some mathematical differences between their models and ours. 
We consider the geometric OU process for a security price.
In contrast, 
Alfonsi et al.~\cite{Alfonsi-Fruth-Schied} and Obizhaeva and Wang~\cite{Obizhaeva-Wang} assume that 
the process of a security price follows arithmetic Brownian motion (or a martingale) and 
that there is exponential (or some more general shape of) resilience for MI in an LOB model. 
The relation between the mean-reverting property of the OU process and the resilience of MI 
causes this interesting similarity in results. 

\subsubsection{The general case}\label{subsec_IS}

In this subsection, we consider the general case, where $\sigma \geq 0$. 
We assume the following condition: 
\begin{eqnarray}\label{cond_phi}
\varphi > \frac{\max\{z, 1 + \beta\}}{\alpha}. 
\end{eqnarray}
This condition means that the amount of the trader's security holdings is larger than that in the case of Section \ref {sec_const}. 
We define the function $H(\lambda ) = H_{t, \varphi }(\lambda )$ on $[0, \infty )$ by 
\begin{eqnarray*}
H(\lambda ) = 
\alpha \exp \left( \alpha \beta \int ^t_0P^{-1}\left( \exp(-e^{-2\beta r}y)\lambda / \alpha \right) dr
- \alpha \varphi + z - y \right) - \lambda. 
\end{eqnarray*}
Note that $H$ is nonincreasing on $[0, \infty)$. 
Moreover, (\ref {cond_phi}) implies that
\begin{eqnarray*}
H\left(\alpha e^{-y}\right) \leq  0 < H(0), 
\end{eqnarray*}
and so the equation $H(\lambda) = 0$ has the unique solution 
$\lambda ^* = \lambda ^*(t, \varphi) \in \left(0, \alpha e^{-y} \right]$. 
The next theorem is the main result in this section. 

\begin{theorem} \ \label{eg_OU}Let $t \in (0, 1], (w, \varphi , s)\in D$ and assume that $(\ref {cond_phi})$ is true. 
Then,
\begin{eqnarray}\label{th_value_fnc}\nonumber 
V_t(w, \varphi ,s ; u_{\mathrm {RN}}) &=& w + 
\frac{s}{\alpha }\left( 1 - \exp\left(-\alpha \varphi + 
\alpha \beta \int ^t_0\xi ^*_rdr\right)\right) \\&& + 
\beta \int ^t_0\xi ^*_r\exp \left (F-\alpha \xi ^*_r + (1 + e^{-2\beta r})y\right)dr, 
\end{eqnarray}
where 
$\xi ^*_r = P^{-1}(\exp(-e^{-2\beta r}y)\lambda^* / \alpha)$. 
\end{theorem}

We can construct a nearly optimal strategy as follows (with $\delta \downarrow 0$): 
\begin{eqnarray}\label{nearly_optimal_strategy}
\hat{\zeta }^{\delta }_r = \frac{p^*}{\delta }1_{[0,\delta ]}(r) + 
\zeta ^*_r + \frac{q^*}{\delta }1_{[t - \delta, t]}(r), 
\end{eqnarray}
where $p^* = \xi ^*_0 + (z - 2y) / \alpha$ and
\begin{eqnarray*}
\zeta^*_r &=& \beta \xi^*_r - 
\frac{2\beta \lambda^*e^{-2\beta r}y\exp (\alpha \xi ^*_r - e^{-2\beta r}y)}{\alpha ^2(2 - \alpha \xi ^*_r)} + 
\frac{2\beta y}{\alpha }e^{-2\beta r}\\
&=& \beta \xi ^*_r + 
\frac{2\beta ye^{-2\beta r}}{\alpha (2 - \alpha \xi^*_r)}, \\
q^* &=& \varphi - \beta \int^t_0\xi ^*_rdr - \xi^*_t - \frac{z}{\alpha } + \frac{y}{\alpha }(1 + e^{-2\beta t}). 
\end{eqnarray*}
Here, the second equality of the definition of $\zeta^*_r$ comes from
$P(\xi^*_r) = \exp (-e^{-2\beta r}y)\allowbreak \lambda^* / \alpha$. 
By the inequalities (\ref {cond_phi}), $z \geq 2y$, and 
$0\leq \xi^*_r \leq \xi^*_0 \leq 1 / \alpha $, 
we see that each of $p^*$, $\zeta^*_r$, and $q^*$ is positive. 

Similarly to in the previous subsection, 
$(\hat{\zeta }^\delta _r)_r$ also consists of three terms: 
initial (almost) block liquidation, 
gradual intermediate liquidation, and terminal (almost) block liquidation. 
However, unlike the case where $\sigma = 0$, 
the speed of execution becomes slower as time passes when $\sigma > 0$. 
This is result is similar to that of Makimoto and Sugihara~\cite{Makimoto-Sugihara}, 
in which the optimal execution problem in an LOB market model with stochastic liquidity is studied. 

Figure \ref {fig_OU} shows a representation of an optimal strategy. 
We can rewrite the value function (\ref {th_value_fnc}) as 
the sum of an initial cash amount and the proceeds of initial/intermediate/terminal liquidation: 
\begin{eqnarray}\nonumber 
&&V_t(w, \varphi ,s ; u_{\mathrm {RN}})\\\label{rw_value_fnc}
&=& w + 
\frac{1 - e^{-\alpha p^*}}{\alpha}s + 
s\int^t_0e^{-\alpha \eta^*_r}\zeta^*_rdr + 
\frac{1 - e^{-\alpha q^*}}{\alpha }se^{-\alpha \eta^*_t}, 
\end{eqnarray}
where $\eta^*_r = \xi^*_r - (1 + e^{-2\beta r})y/\alpha + z/\alpha $. 
Note that $p^* = \eta^*_0$ and 
$\zeta^*_rdr = d\eta^*_r + \beta \xi ^*_rdr$. 

\begin{figure}[t]
\begin{center}
\includegraphics[height = 4.69cm,width=7.296cm]{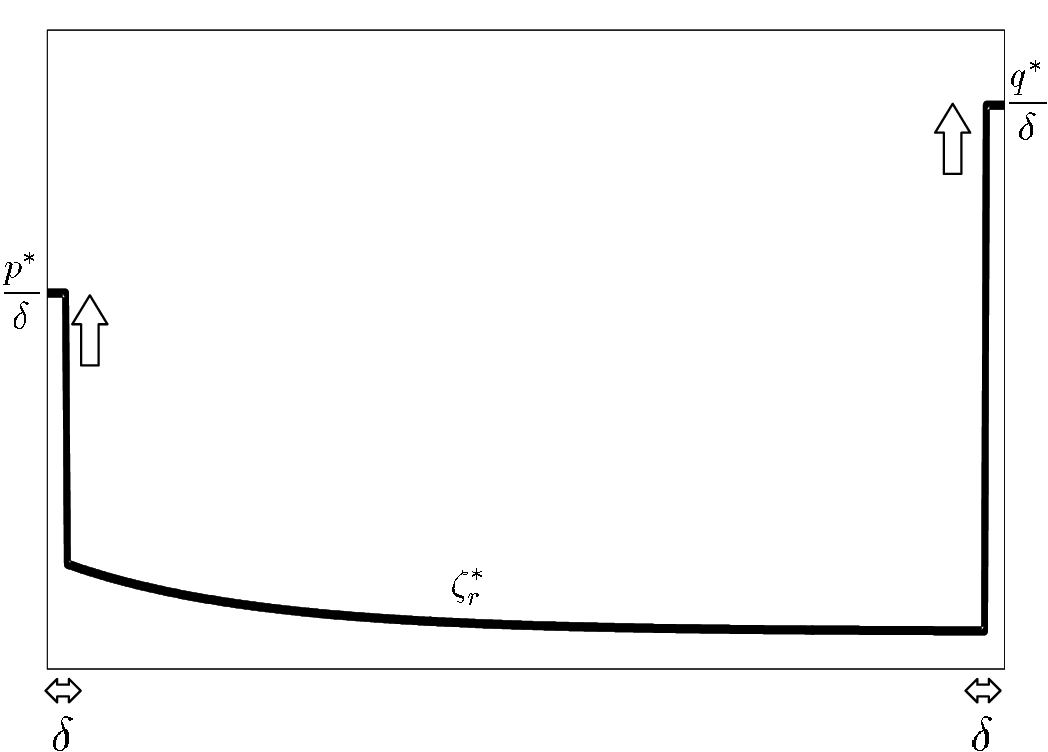}
\includegraphics[height = 4.5cm,width=7cm]{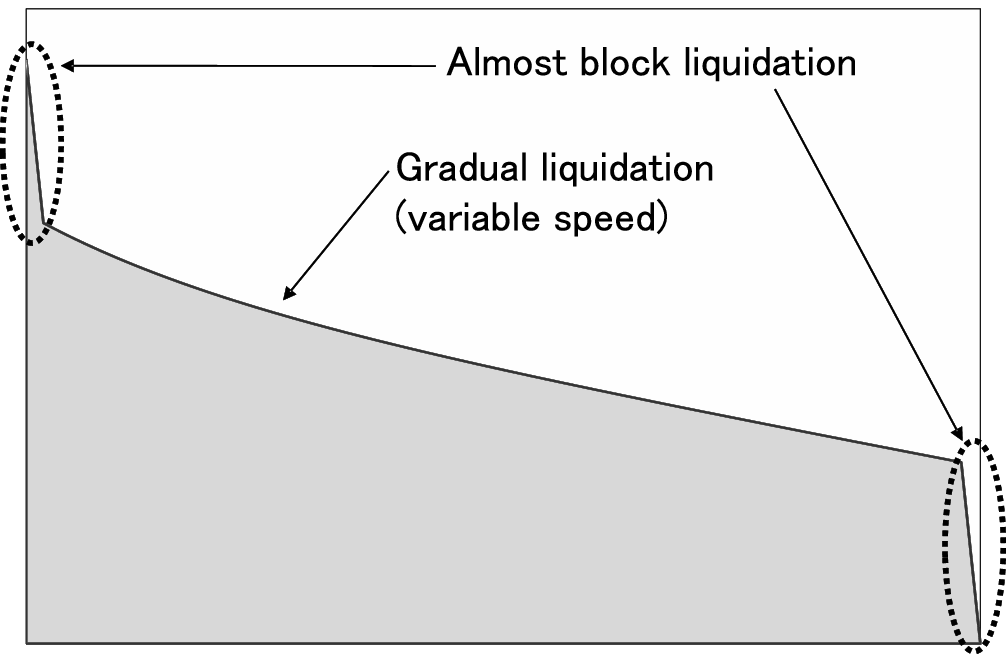}
\caption{Shape of a nearly optimal strategy $(\zeta ^\delta _r)_r$ (left) and 
the corresponding amount of the security holdings (right), both when $\sigma > 0$. 
The horizontal axes are time ($r$). }
\label{fig_OU}
\end{center}
\end{figure}

\subsection{When does MI cause gradual liquidation? Case studies}\label{sec_gradual}

In this subsection, we briefly discuss the properties of MI functions that incentivize 
gradual liquidation of the security and give some examples. 

First, as a benchmark model, we consider the execution problem with risk-neutral trader 
in a Black--Scholes-type model with no MI: 
\begin{eqnarray*}
V^{\mathrm{ideal}}(\Phi_0) &=& \sup_{(\zeta_r)_r\in \mathcal {A}_1(\Phi _0)}
\E[W_1]
\end{eqnarray*}
subject to $(W_0, S_0) = (w, s_0)$, (\ref {fluc_W}), (\ref {exp_X}), and 
\begin{eqnarray}
dX_r = \mu dr + \sigma dB_r, 
\end{eqnarray}
where $\mu \in \mathbb {R}$ and $\sigma \geq 0$ are constants. 
For the sake of simplicity, we assume that $\tilde{\mu } := -\mu + \sigma^2/2 > 0$. 
In this case, an optimal strategy is initial block liquidation: 
since the expected security price decreases as time passes, 
a risk-neutral trader will sell the security as soon as possible. 
Then, we see that $V^{\mathrm {ideal}}(\Phi _0) = \Phi _0s_0$. 

Now, we introduce MI functions. 
We define $V(\Phi _0) = V_1(0, \Phi _0, s_0; u_{\mathrm {RN}})$, where 
the log-price fluctuates according to 
\begin{eqnarray}
dX_r = \mu dr + \sigma dB_r - g(\zeta_r)dr. 
\end{eqnarray}
Note that $V_1(0, \Phi _0, s_0; u_{\mathrm {RN}}) = \hat{V}_1(0, \Phi _0, s_0; u_{\mathrm {RN}})$ holds 
in this case by Proposition 5.1 of \cite {Kato}; $\hat{V}_1$ will be defined in Section \ref {sec_derivation}. 
One of the simplest example is the linear MI function 
$g(\zeta ) = \alpha \zeta $. 
However, Theorem 5.2 in Kato~\cite{Kato} tells us that 
$V(\Phi_0) = (1-\exp (-\alpha \Phi _0))s_0/\alpha$, and 
the optimal strategy is (almost) initial block liquidation (\ref {def_hat_I}). 
This implies that a linear MI is not sufficient to describe 
all situations in which a trader will gradually liquidate. 
On the other hand, Theorem 5.4 of that paper shows that 
the optimal strategy is gradual liquidation when $g(\zeta)$ is quadratic; 
in particular, the optimal strategy with small $\Phi _0$ is 
a time-weighted average price (TWAP) strategy, that is, 
liquidation with constant speed 
(the same result can be obtained when $g$ is S-shaped; see \cite {Kato_JSIAM}). 
Note that TWAP strategies are  equivalent to volume-weighted average price (VWAP) strategies 
with the time parameter $t$ measuring volume time (see Remark 22.7 in \cite {Gatheral-Schied}). 

Here, Theorems \ref {eg_OU_const}--\ref {eg_OU} from the previous subsections give 
another source of gradual liquidation. 
In our model framework, a price-recovery effect also causes gradual intermediate liquidation 
when $\Phi _0$ is large enough, even when the MI function is linear. 

Therefore, the convexity of MI functions and price-recovery effects 
are essential properties in the study of execution problems. 
Note that a price-recovery effect can be identified with market resilience or, 
in other words, a transient MI function. 

\begin{rem}
{\rm Linearity of MI functions is also discussed from another point of view. 
In several studies of execution problems, 
a permanent MI function is often assumed to be linear. 
One of the reasons is that this accords with results from empirical studies. 
For example, Almgren et al.~\cite{Almgren-et-al} estimate the form of MI functions in the real market 
and find that permanent MI function is approximated by a linear function. 
Another reason is the absence of price manipulation: 
Forsyth et al.~\cite{Forsyth-et-al}, 
Gatheral~\cite{Gatheral}, and Huberman and Stanzl~\cite{Huberman-Stanzl1} 
assert that we can construct a price manipulation strategy when the permanent MI function is nonlinear. 
That is, in such cases
a trader can earn money by unscrupulously utilizing the effect of MI. 

In contrast, 
in the LOB markets, nonlinear MI functions are often observed in 
empirical studies 
(see Bouchard et al.~\cite{Bouchaud-Mezard-Potters}, Weber and Rosenow~\cite{Weber-Rosenow}, and 
Zovko and Farmer~\cite{Zovko-Farmer}, for instance). 
Moreover, 
Alfonsi and Schied~\cite{Alfonsi-Schied} and Alfonsi et al.~\cite{Alfonsi-Fruth-Schied, Alfonsi-Schied-Slynko} 
show that the opportunity for price manipulation is not present 
even when the MI function is convex under some technical conditions. 
Therefore, nonlinearity of the MI function is not necessarily inappropriate. 

In this paper, 
we mainly treat the case of a linear MI function $g(\zeta)$, 
but we can also study the case of nonlinear $g(\zeta)$ in a similar way. 
It is meaningful to study the execution problem with both convex MI function $g$ and 
a price-recovery effect in our framework and to investigate topics related to the opportunity for price manipulation 
(the unaffected price, given by (\ref {fluc_X})--(\ref {exp_X}) with $\alpha = 0$, is not a martingale; 
nevertheless, similarity to the results of LOB models may ensure the absence of price manipulation in some sense).}
\end{rem}

\section{Concluding Remarks}\label{sec_conclusion}

In this paper, 
we solved the optimal execution problem 
in the case where the security price follows a geometric Ornstein--Uhlenbeck process and the MI function is linear. 
This case is important because it covers cases where the security price has a mean-reverting property. 
We showed that an optimal static strategy is a mixture of 
initial/terminal block liquidation and gradual intermediate liquidation 
when the initial amount of the security holdings is large. 
When the volatility parameter is equal to zero, 
the optimal strategy has the same form as the strategies found in Alfonsi et al.~\cite{Alfonsi-Fruth-Schied} and 
Obizhaeva and Wang~\cite{Obizhaeva-Wang}. 
In this case, a trader should sell at a constant speed during the time horizon. 
When the volatility is positive, the speed of gradual liquidation is not constant; instead, the speed is decaying and 
the form of the optimal strategy is similar to that found in Makimoto and Sugihara~\cite{Makimoto-Sugihara}. 

Our result demonstrates a case in which MI causes gradual liquidation. 
In a real market, traders sell shares of a security gradually to avoid costs due to MI 
because price recovery is expected. 
As noted in Section \ref {sec_gradual}, 
our results combined with findings from \cite {Kato} imply
that convexity (or nonlinearity) of MI functions and price-recovery effects (or resilience) 
are important factors in the construction of an MI model. 

In Section \ref {sec_gOU}, 
we restrict the class of admissible strategies to static strategies in view of IS algorithms. 
It is obviously important to study the optimal execution problem for adaptive strategies 
($\hat {\mathcal {A}}_t(\varphi)$ in Section \ref {sec_derivation}), which is one of our future tasks.

\appendix 
\section{Appendix}\label{sec_appendix}

\subsection{Derivation of the continuous-time value function}\label{sec_derivation}

In this subsection, 
we study the derivation of the continuous-time model of an optimal execution problem by following the arguments in Appendix A of Kato \cite{Kato}. 
We characterize our 
continuous-time value function 
as a limit of discrete-time models in more general settings. 

We define a discrete-time version of the value function with time intervals of width $1/n$ as follows: 
\begin{eqnarray}\label{def_value_discrete}
\hat{V}^n_k(w,\varphi ,s ; u) = \sup_{(\psi^n_l)^{k-1}_{l=0}\in \hat{\mathcal {A}}^n_k(\varphi)}
\E[u(W^n_k,\varphi ^n_k, S^n_k)]
\end{eqnarray}
subject to 
\begin{eqnarray}
\label{W_varphi}
&&W^n_{l+1} = W^n_l+\psi^n_lS^n_l\exp (-g_n(\psi^n_l)), \ \ 
\varphi ^n_{l+1} = \varphi^n_l - \psi^n_l, \\
\label{fluctuate_X}
&&X^n_{l+1} = Y\Big(\frac{l+1}{n} ; \frac{l}{n}, X^n_l-g_n(\psi^n_l)\Big), 
\ \ S^n_{l+1} = \exp (X^n_{l+1}) 
\end{eqnarray}
and $(W^n_0, \varphi^n_0, S^n_0) = (w, \varphi , s)$, where 
$u(w, \varphi , s)$ is a nondecreasing and continuous function with polynomial growth rate in each of $w, \varphi$, and $s$, 
$g_n: [0,\infty )\longrightarrow [0,\infty )$ is a nondecreasing and continuously 
differentiable function that satisfies $g_n(0) = 0$, 
$Y(t ; r, x)$ is a solution of the stochastic differential equation 
\begin{eqnarray}\label{SDE_Y}
\left\{
\begin{array}{ll}
 	dY(t; r,x) = \sigma (Y(t; r,x))dB_t+b(Y(t; r,x))dt,& t\geq r,	\\
 	\hspace{2mm}Y(r; r,x) = x,&
\end{array}
\right.
\end{eqnarray}
and $b, \sigma  : \mathbb {R}\longrightarrow \mathbb {R}$ are Lipschitz-continuous functions. 
Here, $\hat{\mathcal {A}}^n_k(\varphi )$ is the set of stochastic processes 
$(\psi ^n_l)^{k-1}_{l=0}$ such that 
$\psi ^n_l$ is $\mathcal {F}_{l/n}$-measurable, $\psi ^n_l\geq 0$ for each $l=0,\ldots ,k-1$, 
and $\sum ^{k-1}_{l=0}\psi ^n_l\leq \varphi $. 
See \cite {Kato} for the precise definitions and financial implications of (\ref {def_value_discrete}). 
Note that $b$ and $\sigma$ are assumed to be bounded in \cite{Kato}, 
but here we only assume that 
$b$ and $\sigma$ have linear growth. Because of this, the OU case (\ref {fluc_X}) is included in our settings 
(i.e., uniqueness and existence of the solution $Y(\cdot ; r, x)$ is guaranteed for each $r\geq 0$ and $x\in \mathbb {R}$). 

We consider the limit as $n\rightarrow \infty $. 
Let $h : [0,\infty) \longrightarrow [0,\infty)$ be a nondecreasing continuous function. 
We quote the condition assumed in \cite{Kato}: \vspace{2mm}\\
$[A]$ \ 
$\lim _{n\rightarrow \infty }\sup _{\psi \in [0,\Phi _0]}
\Big |\frac{d}{d\psi }g_n(\psi )-h(n\psi )\Big | = 0$. \vspace{2mm}
\\
Now, we define $g(\zeta ) = \int ^\zeta_0h(\zeta')d\zeta'$ for $\zeta \in [0,\infty)$. 
This function corresponds to an MI function in a continuous-time model. 
We define 
\begin{eqnarray}\label{def_value_conti}
\hat{V}_t(w,\varphi ,s ; u) = \sup _{(\zeta _r)_{r}\in \hat{\mathcal {A}}_t(\varphi )}
\E [u(W_t,\varphi _t, S_t)] 
\end{eqnarray}
subject to (\ref {fluc_W})--(\ref {fluc_X_general}) and 
$\varphi _t = \varphi - \int ^t_0\zeta _vdv$. 
Here, $\hat{\mathcal {A}}_t(\varphi )$ denotes the set of 
admissible adaptive strategies, that is, 
$(\zeta _r)_{0\leq r\leq t}\in \hat{\mathcal {A}}_t(\varphi )$ is 
an $(\mathcal {F}_r)_r$-progressively measurable process such that 
\begin{itemize}
 \item [(a'.)] $\zeta _r\geq 0$ for each $r\in [0, T]$ almost surely, \ 
 \item [(b'.)] $\int ^T_0\zeta _rdr\leq \Phi _0$ almost surely, and 
 \item [(c'.)] $\sup _{r,\omega }\zeta _r(\omega )<\infty$. 
\end{itemize}
Note that condition (c'.) is a technical condition; see Appendix A of \cite {Kato} for the details. 
Here, we make a further assumption: \vspace{2mm}\\
$[B]$ \ For each $m \in \mathbb {N}$, 
there is a constant $C_m > 0$ such that 
$\E[\sup_{0\leq t\leq 1}\exp (mY(t ; 0, x))]\leq C_me^{mx}$. 

This condition does not seem to be natural in general, 
but our main model in Section \ref{sec_model} satisfies [B]. 

The following is a generalization of Theorem A.1 of \cite{Kato}. 

\begin{theorem} \ \label{converge}Assume that $\mathrm {[A]}$--$\mathrm {[B]}$ hold. 
For each $(w,\varphi,s)\in D$, $t\in [0,1]$ and $u\in \mathcal {C}$,
\begin{eqnarray}\label{th_2}
\lim_{n\rightarrow \infty }\hat{V}^n_{[nt]}(w,\varphi,s;u) = \hat{V}_t(w,\varphi,s;u), 
\end{eqnarray}
where $[nt]$ is the greatest integer less than or equal to $nt$. 
\end{theorem}

By this theorem, we can characterize (\ref{def_value_conti}) as 
the limit of a sequence of discrete-time value functions (\ref {def_value_discrete}). 
Note that we can also obtain a similar convergence result 
by restricting the classes of admissible strategies to classes of static strategies 
(Proposition \ref {prop_conv_fn} below). 

To prove Theorem \ref {converge}, we introduce the following lemma: 

\begin{lemma}\label{cond_of_X1} \ Let $t\in [0,1]$, $\varphi \geq 0$,  
$(\zeta _r)_{0\leq r\leq t}\in \hat{\mathcal {A}}_t(\varphi)$, and let 
$(X_r)_{0\leq r\leq t}$ be given by $(\ref{fluc_X_general})$. 
Then, there is a constant $C>0$, depending on only $b$ and $\sigma$, such that 
\begin{eqnarray}&&\nonumber 
\E \Big [\sup _{r\in [r_0,r_1]}
\Big |X_r-X_{r_0}+\int ^r_{r_0}g(\zeta _v)dv\Big |^4\Big ] \\\label{temp_X1}
&\leq & 
C(r_1-r_0)^2\{ 1 + (r_1-r_0)^3\int ^{r_1}_{r_0}\E [g(\zeta _v)^4]dv\}
\end{eqnarray}
for each $0\leq r_0\leq r_1\leq t$. 
\end{lemma}

The proof of Theorem \ref {converge} is given in almost the same way as the proofs of Propositions B.24--B.25, 
but we apply condition [B] and Lemma \ref {cond_of_X1} instead of Lemmas B.1 and B.3, respectively, from \cite{Kato}. 

Unlike the case where $b$ is bounded, the right hand side of (\ref {temp_X1}) depends on $(\zeta _r)_r$. 
However, this does not change the essential method of proving 
the statements analogous to Theorems 3.1--3.2 and 4.1--4.2 of \cite{Kato}, 
except for the continuity of the continuous-time value function at $t = 0$ when $h(\infty ) = \infty$. 
We omit the details here.

\subsection{Proofs}\label{sec_proofs}

For brevity, we assume $t = 1$ until the end of this section. 
We define 
\begin{eqnarray}\label{def_f_n}
f^n &=& 
\frac{1}{\alpha}\sup _{(\psi^n_k)_k\in \mathcal {A}^{n}_n(\varphi )}
\tilde{f}^n(\psi^n_0, \ldots, \psi^n_{n-1}). 
\end{eqnarray}
Here, $\mathcal {A}^{n}_n(\varphi)$ is the set of admissible deterministic strategies. 
That is, 
\begin{eqnarray*}
\mathcal {A}^{n}_n(\varphi ) = 
\left\{(x_l)^{n-1}_{l=0}\subset [0, \varphi]^n\ ; \ \sum ^{n-1}_{l = 0}x_l\leq \varphi\right\}, 
\end{eqnarray*}
$\tilde{f}^n(x)$, $x = (x_0, \ldots , x_{n-1})\in \mathbb {R}^n$, is defined by 
\begin{eqnarray*}
\tilde{f}^n(x)
&=& 
\alpha \sum ^{n-1}_{k=0}
\exp \left(c^k_nz - c^{2k}_ny - \alpha \sum ^{k-1}_{l = 0}c^{k-l}_nx_l\right )
\int^{(k+1)/n}_{k/n}nx_k\exp (-\alpha (nr - k)x_k)dr\\
&=& 
\sum ^{n-1}_{k=0}
\exp \left (c^k_nz - c^{2k}_ny - \alpha \sum ^{k-1}_{l = 0}c^{k-l}_nx_l\right)
(1-e^{-\alpha x_k}), 
\end{eqnarray*}
and $c_n = e^{-\beta /n}$. 
Since the function $\tilde{f}^n(x_0, \ldots , x_{n-1})$ is nondecreasing in $x_{n-1}$, 
we can replace $\mathcal {A}^{n}_n(\varphi )$ in (\ref {def_f_n}) with 
\begin{eqnarray*}
\mathcal {A}^{n, \mathrm {SO}}_n(\varphi) = 
\left\{ (x_l)^{n-1}_{l=0}\subset [0, \Phi _0]^n\ ; \ \sum ^{n-1}_{l = 0}x_l = \varphi\right\} 
\subset \mathcal {A}^{n}_n(\varphi). 
\end{eqnarray*}

The following proposition holds. 

\begin{proposition} \ \label{prop_conv_fn} 
\begin{itemize}
 \item [$\mathrm {(i)}$] $f^n\longrightarrow f$, \ \ $n\rightarrow \infty $, where 
\begin{eqnarray}\label{def_f}
f &=& \sup _{(\zeta _r)_r\in \mathcal {A}_1(\varphi )}\tilde {f}((\zeta _r)_r), \\\nonumber 
\tilde{f}((\zeta _r)_r) &=& \int ^1_0\zeta _r\exp \left( e^{-\beta r}z  - e^{-2\beta r}y - 
\alpha e^{-\beta r}\int ^r_0e^{\beta v}\zeta _vdv\right)dr, 
\end{eqnarray}
 \item [ $\mathrm {(ii)}$ ] $V_1(w, \varphi , s ; u_{\mathrm {RN}}) = w + e^{F+y}f$. 
\end{itemize}
\end{proposition}

\begin{proof}
An argument analogous to the proof of Proposition B.24--25 in Kato~\cite{Kato} yields assertion (i). 
Assertion (ii) can be obtained by straightforward calculation. 
\end{proof}

\subsubsection{Proof of Theorems \ref {th_no_MI}--\ref {th_phi_small}}\label{sec_proof0}

We give the proof of only Theorem \ref {th_phi_small} 
because the proof of Theorem \ref {th_no_MI} is essentially the same. 
By straightforward calculation, we obtain 
\begin{eqnarray*}
f \geq \lim _{\delta \rightarrow 0}\tilde{f}((\hat{\zeta }^{0,\delta }_r)_r) = 
\frac{1 - e^{-\alpha \varphi }}{\alpha}e^{z-y}, 
\end{eqnarray*}
where $(\hat{\zeta }^{0,\delta }_r)_r$ is defined as (\ref {def_hat_I}). 
On the other hand, 
for any $(\zeta _r)_r\in \mathcal {A}_1(\varphi)$, we have 
\begin{eqnarray*}
\tilde{f}((\zeta _r)_r) \leq  
\int ^1_0\zeta _r\exp \left( e^{-\beta r}z  - e^{-2\beta r}y - 
\alpha e^{-\beta r}\eta _r\right)dr, 
\end{eqnarray*}
where $\eta _r = \int ^r_0\zeta _vdv$. 
From the relation $z - 2y \geq \alpha \varphi \geq \alpha \eta_r$, we have that 
\begin{eqnarray*}&&
\{ z - y - \alpha \eta _r\} - 
\{ e^{-\beta r}z  - e^{-2\beta r}y - 
\alpha e^{-\beta r}\eta _r \} \\
&=& 
(1-e^{-\beta r})(z-(1+e^{-\beta r})y - \alpha \eta_r) \geq 0. 
\end{eqnarray*}
Thus, 
\begin{eqnarray*}
\tilde{f}((\zeta_r)_r)\leq \int ^t_0\exp (z - y - \alpha \eta_r)d\eta _r \leq 
\frac{1 - e^{-\alpha \varphi}}{\alpha}e^{z - y}. 
\end{eqnarray*}
Therefore, we get $f\leq (1 - e^{-\alpha \varphi })e^{z-y} / \alpha$, 
and this completes the proof of Theorem \ref {th_phi_small}. \qed

\subsubsection{Proof of Theorems \ref {eg_OU_const}--\ref {eg_OU}}\label{sec_proof}

Because Theorem \ref {eg_OU_const} is a corollary of Theorem \ref {eg_OU} 
(a small technical argument is necessary to generalize the assumption (\ref {cond_phi})), 
we present the proof of only Theorem \ref{eg_OU}. 

Let $\Xi^n(\varphi ) = \{ (x_0, \ldots, x_{n-1})\in \mathbb {R}^n\ ; \ x_0 + \cdots + x_{n-1} = \varphi\}$. 
Note that $\mathcal {A}^{n, \mathrm {SO}}_k(\varphi ) \subset \Xi ^n(\varphi )$. 
We set $\tilde {Q}^n_k(x) = \sum ^{l}_{m = 0}c^{l-m}_nx_m$ and 
$Q^n_k(x) = -zc^k_n + yc^{2k}_n + \alpha \tilde {Q}^n_k(x)$. 

\begin{lemma} \ \label{conv_Q}
$\min_{k = 0, \ldots , n-1 }Q^n_k(x)\ \longrightarrow \ -\infty $ as $|x|\rightarrow \infty $ on 
$\Xi^n(\varphi )$. 
\end{lemma}

\begin{proof} 
It suffices to show that $\min_{k = 0, \ldots , n-1 }\tilde{Q}^n_k(x)\ \longrightarrow \ -\infty $. 
Choose any $M > 0$. 
Let $x\in \Xi ^n(\varphi)$ be such that 
$\min_{k = 0, \ldots, n-1}\tilde{Q}^n_k(x) \geq -M$. 
Then, it holds that 
\begin{eqnarray}\label{temp_est_x1}
x_k + c_nx_{k-1} + \cdots + c^k_nx_0 \geq -M, \ \ k = 0, \ldots, n - 1. 
\end{eqnarray}
From (\ref {temp_est_x1}) with $k = n - 1$ and the equality $x_{n-1} + \cdots + x_0 = \varphi$, 
we observe 
\begin{eqnarray}\label{temp_est_x2}
\sum^{n - 2}_{k = 0}\left( \sum^{n - 2 - k}_{l = 0}c^l_n \right) x_k 
\leq \frac{M + \varphi }{1 - c_n}. 
\end{eqnarray}
By (\ref {temp_est_x2}) and (\ref {temp_est_x1}) with $k = n - 2$, we have 
\begin{eqnarray*}
\sum ^{n - 3}_{k = 0}\left(\sum^{n - 3 - k}_{l = 0}c^l_n\right) x_k \leq 
\left(\frac{1}{1 - c_n} + 1\right) (M + \varphi). 
\end{eqnarray*}
Inductively, 
\begin{eqnarray}\label{temp_est_x3}
\sum^{k}_{k' = 0}\left(\sum^{k - k'}_{l = 0}c^l_n\right) x_{k'} \leq 
\left(\frac{1}{1 - c_n} + n - 2 - k\right) (M + \varphi) \leq a_n(M + \varphi) 
\end{eqnarray}
for $k = 0, \ldots , n - 2$, where $a_n = \{ (1 - c_n)^{-1} + n\}$. 

By (\ref {temp_est_x1}) and (\ref {temp_est_x3}) with $k = 0$, we have $-M\leq x_0\leq a_n(M + \varphi)$. 
Similarly, by (\ref {temp_est_x1}) and (\ref {temp_est_x3}) with $k = 1$, we have 
$-(1 + a_nc_n)(M + \varphi) \leq x_1 \leq (a_n + 1 + c_n )(M + \varphi)$. 
By an inductive calculation, we arrive at $|x_k|\leq C_n(M + \varphi), k = 0, \ldots , n - 2$ for some $C_n > 0$. 
Combining this with the relation $x\in \Xi ^n(\varphi)$, we also see that $|x_{n-1}|\leq C'_n(M + \varphi)$ 
for some $C'_n > 0$. 

The above arguments tell us the following: 
if a sequence $(x^{(N)})_N\subset \Xi ^n(\varphi )$ satisfies \allowbreak $\lim _{N\rightarrow \infty }\min _k\allowbreak \tilde{Q}^n_k(x^{(N)}) \neq -\infty $, 
then $(x^{(N)})_N$ is bounded. 
The desired assertion follows by contradiction. 
\end{proof}

\begin{lemma} \ \label{lemma_divergence}
$\tilde{f}^n(x_0, \ldots , x_{n-1}) \longrightarrow -\infty$ as 
$|x| \rightarrow \infty $ on $\Xi ^n(\varphi )$. 
\end{lemma}

\begin{proof} 
Let $A_n(p) = e^{-c_np + y} - e^{-p}$, $p\in \mathbb {R}$. This gives 
\begin{eqnarray*}
\tilde{f}^n(x) &=& 
\sum ^{n-1}_{k=0}(e^{-c_nQ^n_{k-1}(x) + yc^{2k-1}_n(1-c_n)} - e^{-Q^n_k(x)})\\
&\leq & 
e^{z-y} - e^{-Q^n_{n-1}(x)} + \sum ^{n-2}_{k = 0}A_n(Q^n_k(x)) 
\end{eqnarray*}
for any $x = (x_0, \ldots , x_{n-1})\in \mathbb {R}^n$, 
where $Q^n_{-1}(x) = -c^{-1}_nz + c^{-2}_ny$. 
We easily see that the function $A_n$ has an upper bound $C_{A, n}$. 
Hence, it holds that 
\begin{eqnarray*}
\tilde{f}^n(x) \leq 
e^{z-y} - \exp (-\min _kQ^n_k(x)) + 
A_n(\min _kQ^n_k(x)) + C_{A, n}(n-1). 
\end{eqnarray*}
Since $\lim _{p\rightarrow -\infty }A_n(p) = -\infty$, we have the desired assertion by the above inequality and Lemma \ref {conv_Q}. 
\end{proof}

A straightforward calculation gives 
\begin{lemma} \ \label{lem_eg_1}For each $k = 0, \ldots , n-2$, it holds that 
\begin{eqnarray*}
&&\frac{\partial}{\partial x_k}\tilde{f}^n(x_0, \ldots , x_{n-1})\\&=& 
c_n\frac{\partial}{\partial x_{k+1}}\tilde{f}^n(x_0, \ldots , x_{n-1}) + 
\alpha (1-c_n)\exp(-c^{2k}_ny)F^n_k\left (\sum ^k_{l = 0}c^{k - l}_nx_l - c^k_nz/\alpha \right), 
\end{eqnarray*}
where 
\begin{eqnarray}\label{eq_gamma}
F^n_k(x) &=& 
e^{-\alpha x}\left( 
1 + c_n\cdot \frac{1 - e^{\tilde{c}_n\gamma ^n_k(x)}}{\tilde{c}_n}
\right), \\
\gamma ^n_k(x) &=& \alpha x + (1+c_n)c^{2k}_ny, \ \ \tilde{c}_n = 1 - c_n. 
\end{eqnarray}
\end{lemma}
Formally, the Taylor expansion of $F^n_{[nr]}(x)$ is given by 
\begin{eqnarray*}
F^n_{[nt]}(x) &=& e^{-\alpha x}\left\{ 1 + c_n\gamma ^n_{[nt]}(x) + O(n^{-1})\right\} \\
&\longrightarrow &
e^{-\alpha x}\left\{ 1 + \alpha x + (1 + 1)e^{-\beta t}y \right\} = 
P(x) + 2e^{-\alpha x - \beta t}y, \ \ n\rightarrow \infty. 
\end{eqnarray*}
The following lemma gives higher-order estimates of the above calculation. 
\begin{lemma} \ \label{conv_unif_F}It holds that 
\begin{eqnarray*}
\max_{k = 0, \ldots , n - 1}
\sup_{x\in K}
\left| n(F^n_k(x) - P(x) + 2e^{-\alpha x }c^{2k}_ny) - G^n_k(x) 
\right| \longrightarrow 0, \ \ n\rightarrow \infty 
\end{eqnarray*}
for each compact set $K \subset \mathbb {R}$, where, 
\begin{eqnarray*}
G^n_k(x) &=& \beta e^{-\alpha x}(\alpha x + (2 + c_n)c^{2k}_ny - c_nR^n_k(x)), \\
R^n_k(x) &=& \int ^1_0\exp (v(1 - c_n)\gamma^n_k(x))(1-v)dv(\gamma^n_k(x))^2. 
\end{eqnarray*}
\end{lemma}

\begin{proof} 
Using (\ref {eq_gamma}) and Taylor's theorem, we have 
\begin{eqnarray*}
F^n_k(x) &=& 
e^{-\alpha x}\left\{ 
1 - c_n(\gamma ^n_k(x) - \tilde{c}_nR^n_k(x))
\right\} \\&=& 
P(x) - 2e^{-\alpha x }c^{2k}_ny + 
\tilde{c}_nG^n_k(x) / \beta. 
\end{eqnarray*}
Thus, it holds that 
\begin{eqnarray*}
\left| n(F^n_k(x) - P(x) + 2e^{-\alpha x }c^{2k}_ny) - G^n_k(x) \right| 
\leq 
|n\tilde{c}_n / \beta - 1|\cdot |G^n_k(x)|. 
\end{eqnarray*}
Since we have $n\tilde{c}_n \longrightarrow \beta $ as $n\rightarrow \infty $ and 
\begin{eqnarray}\label{est_Gnk}
|G^n_k(x)| \leq 2\beta e^{2\alpha |x| + 2y}(\alpha |x| + \alpha ^2|x|^2 + 3y + 4y^2), 
\end{eqnarray}
we have obtained the desired assertion. 
\end{proof}

Because $F^n_k$ is nonincreasing on $E^n_k$, 
we can define the (nonincreasing) inverse function $F^{n, -1}_k$ on $[0, \infty)$, where 
\begin{eqnarray*}
E^n_k = \left (-\infty , -\frac{1}{\alpha}\left(c^{2k}_n(c_n + 1)y + \frac{\log c_n}{1 - c_n}\right) \right]. 
\end{eqnarray*}
We consider an analogous approximation of $F^{n, -1}_k$, such as Lemma \ref{conv_unif_F}. 
For this, we let 
\begin{eqnarray*}
I(q) &=& \frac{d}{dq}P^{-1}(q) = \frac{\exp (\alpha P^{-1}(q))}{\alpha (\alpha P^{-1}(q) - 2)}, \\
J^n_k(q) &=& -\exp (-2c^{2k}_ny)I(\exp (-2c^{2k}_ny)q)G^n_k(F^{n, -1}_k(q)), \\
\varepsilon ^n_k(q) &=&  F^{n,-1}_k(q) - P^{-1}(\exp (-2c^{2k}_ny)q) + 2c^{2k}_ny / \alpha . 
\end{eqnarray*}

\begin{lemma} \ \label{lem_conv_F_inv}It holds that \\
$\mathrm {(i)}$ \ 
$\max _{k = 0, \ldots , n - 1}\sup _{0\leq q\leq M}
\left| \varepsilon ^n_k(q)\right| \longrightarrow 0$, \\
$\mathrm {(ii)}$ \ 
$\max _{k = 0, \ldots , n - 1}\sup _{0\leq q\leq M}
\left| n\varepsilon ^n_k(q) - J^n_k(q) \right| 
\longrightarrow 0$\\
as $n\rightarrow \infty $ for each $M > 0$. 
\end{lemma}

\begin{proof} 
Assertion (i) is a direct consequence of assertion (ii), 
so we will prove only (ii). 
Take any $q\in [0, M]$ and let $x^n_k = F^{n, -1}_k(q)$. 
Since $F^n_k(x)$ is nondecreasing with respect to $n$ and $k$ for each fixed $x$, 
we get $x^n_k\in K_M$ for any $n$ and $k$, where 
\begin{eqnarray*}
K_M = \left[ F^{1,-1}_0(M), \frac{\beta }{\alpha (1 - e^{-\beta })}\right]. 
\end{eqnarray*}
Let $\tilde{R}^n_k(x) = F^n_k(x) - P(x) + 2e^{-\alpha x }c^{2k}_ny$. 
From the relation
\begin{eqnarray*}
P(x^n_k) - 2e^{-\alpha x^n_k}c^{2k}_ny + \tilde{R}^n_k(x^n_k) = q, 
\end{eqnarray*}
we have 
\begin{eqnarray}
P(x^n_k + 2c^{2k}_ny / \alpha ) = \exp (-2c^{2k}_ny)(q - \tilde{R}^n_k(x^n_k)). 
\end{eqnarray}
Applying $P^{-1}$ to both sides and subtracting $P^{-1}(\hat{e}^n_kq)$, we obtain 
\begin{eqnarray*}
\varepsilon ^n_k(q) = 
P^{-1}(\hat{e}^n_k(q - \tilde{R}^n_k(x^n_k))) - P^{-1}(\hat{e}^n_kq), 
\end{eqnarray*}
where we denote $\exp (-2c^{2k}_ny)$ as $\hat{e}^n_k$  for brevity. 
Therefore, 
\begin{eqnarray}\nonumber 
\left| n\varepsilon ^n_k(q) - J^n_k(q) \right| 
&\leq &
\left| 
-\int ^1_0
I(\hat{e}^n_k(q - v\tilde{R}^n_k(x^n_k)))dv
\hat{e}^n_kn\tilde{R}^n_k(x^n_k) + \hat{e}^n_kI(\hat{e}^n_kq)G^n_k(x^n_k)
\right| \\\nonumber 
&\leq & 
\sup _{v\in [0, 1]}|I(\hat{e}^n_k(q-v\tilde{R}^n_k(x^n_k)))|\cdot |n\tilde{R}^n_k(x^n_k) - G^n_k(x^n_k)|\\
&& + 
|G^n_k(x^n_k)|\int ^1_0|I(\hat{e}^n_k(q-v\tilde{R}^n_k(x^n_k))) - I(\hat{e}^n_kq)|dv. 
\label{temp_calc}
\end{eqnarray}
Because Lemma \ref{conv_unif_F} implies 
\begin{eqnarray}\label{conv_Rnk}
\max _{k = 0, \ldots, n-1}\sup _{x\in K_M\cap E^n_k}|\tilde{R}^n_k(x)| \longrightarrow 0, \ \ 
n\rightarrow \infty, 
\end{eqnarray}
we can see that 
\begin{eqnarray}\label{ineq_hat_e}
\hat{e}^n_k(q - \tilde{R}^n_k(x^n_k)) > -e^{-3/2}/2 > -e^{-2}
\end{eqnarray}
for large enough $n$ and $k = 0, \ldots, n - 1$. 
Combining the inequality 
\begin{eqnarray*}
-\frac{2e^{3/2}}{\alpha } \leq I(q) < 0 < \frac{d}{dq}I(q) \leq \frac{12e^3}{\alpha}, \ \ 
x \geq -e^{-3/2}/2
\end{eqnarray*}
with (\ref {temp_calc}) and (\ref {ineq_hat_e}), we obtain 
\begin{eqnarray*}
\left| n\varepsilon ^n_k(q) - J^n_k(q) \right| 
\leq 
\frac{2e^{3/2}}{\alpha }
|n\tilde{R}^n_k(x^n_k) - G^n_k(x^n_k)| + 
\frac{12e^3}{\alpha }|\tilde{R}^n_k(x^n_k)|\cdot |G^n_k(x^n_k)|. 
\end{eqnarray*}
Now, we complete the proof of the assertion (ii) by combining 
the above inequality with 
Lemma \ref {conv_unif_F}, (\ref {est_Gnk}), and (\ref {conv_Rnk}). 
\end{proof}

Note that 
Lemma \ref {lem_conv_F_inv} guarantees that the following expansion is valid and that convergence occurs: 
\begin{eqnarray*}
F^{n, -1}_{[nr]}(q) &=& P^{-1}(\exp (-2c^{2[nr]}_ny)q) - 2c^{2[nr]}_ny / \alpha + 
\frac{1}{n}J^n_{[nr]}(q) + o(n^{-1})\\
&\longrightarrow &
P^{-1}\left( \exp (-2e^{-2\beta r}y)q\right) - 2e^{-2\beta r}y / \alpha , \ \ n\rightarrow \infty. 
\end{eqnarray*}

Here, we study a discretization of the function $H(\lambda)$. 
Let 
\begin{eqnarray*}
H_n(\lambda ) = 
\alpha \exp \left(\alpha (1 - c_n)\sum ^{n - 2}_{k = 0}F^{n, -1}_k(\exp (c^{2k}_ny)\lambda / \alpha) - 
\alpha \varphi + z - c^{2(n-1)}_ny\right) - \lambda. 
\end{eqnarray*}
Lemma \ref {lem_conv_F_inv} immediately implies the following proposition. 

\begin{proposition} \ \label{prop_conv_unif_H}
$H_n$ converges uniformly to $H$ on any compact subset of $\mathbb{R}$. 
\end{proposition}

By Proposition \ref {prop_conv_unif_H} and the fact that $H_n$ is strictly decreasing on $[0, \infty)$, 
we can choose an $n$ large enough that 
there is a unique solution $\hat{\lambda }^n$ of $H_n(\lambda ) = 0$ on $(0, 2\lambda ^*)$. 
Moreover, it follows that $\hat{\lambda}^n$ converges to $\lambda^*$ 
as $n\rightarrow \infty$. 

From now, we construct an optimizer to $f^n$ for large enough $n$. 
Set $\hat{\psi }^n_k = \mathcal {T}_k(\hat{\lambda }^n)$, $k = 0, \ldots, n-1$, where 
\begin{eqnarray*}\label{eq_relation_psi_lambda}
\left. 
\begin{array}{ccl}
&&\mathcal {T}_0(\lambda ) = F^{n, -1}_0(\exp (y)\lambda / \alpha) + z / \alpha , \\
&&\mathcal {T}_k(\lambda ) = F^{n, -1}_k(\exp (c^{2k}_ny)\lambda / \alpha) - 
c_nF^{n,-1}_{k - 1}(\exp (c^{2(k - 1)}_ny)\lambda / \alpha), \ k = 1, \ldots , n-2, \\
&&\mathcal {T}_{n-1}(\lambda) = 
\varphi - (1 - c_n)\sum ^{n - 3}_{k = 0} F^{n,-1}_k(\exp (c^{2k}_ny)\lambda / \alpha)\\&&\hspace{16mm} - 
F^{n,-1}_{n-2}(\exp (c^{2(n - 2)}_ny)\lambda / \alpha ) - z/\alpha. 
\end{array}
\right. 
\end{eqnarray*}

The following lemma is obtained from Lemma \ref {lem_conv_F_inv}, the convergences 
$n(1-c_n)\longrightarrow \beta$ and $\hat{\lambda }^n\longrightarrow \lambda^*$ as 
$n\rightarrow \infty$, and the finiteness of 
\begin{eqnarray*}
\sup_{n, k}\sup _{q\in K\cap E^n_k}|J^n_k(q)| < \infty 
\end{eqnarray*}
for each compact subset $K$ of $\mathbb{R}$. 

\begin{lemma} \ \label{conv_str}It holds that 
\begin{eqnarray*}
|\hat{\psi }^n_0 - p^*| + \max _{k = 1, \ldots , n - 2}|n\hat{\psi }^n_k - \zeta ^*_{k/n}| + 
|\hat{\psi }^n_{n-1} - q^*| \ \longrightarrow \ 0, \ \ n\rightarrow \infty . 
\end{eqnarray*}
\end{lemma}

Lemma \ref {conv_str} and the relations $p^*$, $\zeta^*_r$, and $q^* > 0$ together imply the following lemma. 

\begin{lemma} \ \label{psi_positive} It holds that 
$\hat{\psi }^n_k > 0$, \ $k = 0, \ldots , n - 1$ for sufficiently large values of $n$; 
thus, $(\hat{\psi }^n_k)_k\in \mathcal {A}^{n, \mathrm {SO}}_n(\varphi )$. 
\end{lemma}

Now, we define an $(n+1)$-variable function $\mathcal {L}_n(x_0, \ldots, x_{n-1}, \lambda)$ by 
\begin{eqnarray*}
\mathcal {L}_n(x_0, \ldots , x_{n-1}, \lambda) = 
\tilde{f}^n(x_0, \ldots , x_{n-1}) + \lambda(\varphi - x_0 - \cdots - x_{n-1}). 
\end{eqnarray*}
Then, we have the following. 

\begin{lemma} \ \label{lemma_Lagrange}
When $n$ is large enough, 
a solution of 
\begin{eqnarray}\label{eg_temp_eqn1}
\frac{\partial }{\partial x_0}\mathcal {L}_n = \cdots = 
\frac{\partial }{\partial x_{n-1}}\mathcal {L}_n = 
\frac{\partial }{\partial \lambda }\mathcal {L}_n = 0. 
\end{eqnarray}
coincides with $(\hat{\psi }^n_0, \ldots , \hat{\psi }^n_{n-1}, \hat{\lambda }^n)$. 
\end{lemma}

\begin{proof} 
Suppose that a vector $(\tilde{x}_0, \ldots, \tilde{x}_{n-1}, \tilde{\lambda })$ 
is a solution of (\ref {eg_temp_eqn1}). 
Then, we have that $\tilde{x}_0 + \cdots + \tilde{x}_{n -1} = \varphi$, and Lemma \ref {lem_eg_1} implies 
\begin{eqnarray*}
\tilde{\lambda } = c_n\tilde{\lambda } + 
\alpha (1-c_n)\exp (-c^{2k}_ny)F^n_k\left (\sum ^k_{l = 0}c^{k - l}_n\tilde{x}_l - c^k_nz/\alpha \right); 
\end{eqnarray*}
hence, 
\begin{eqnarray}\label {temp_lemma_eq_3}
\sum ^k_{l = 0}c^{k - l}_n\tilde{x}_l = 
F^{n, -1}_k(\exp (c^{2k}_ny)\tilde{\lambda}/\alpha) + c^k_nz/\alpha, \ \ k = 0, \ldots, n - 2. 
\end{eqnarray}
Then, we see that $\tilde{x}_k = \mathcal {T}_k(\tilde{\lambda})$, $k = 0, \ldots, n - 1$. 
Therefore, 
\begin{eqnarray*}\label{eg_temp_eqn2}
0 &=& \frac{\partial}{\partial x_{n-1}}\mathcal {L}_n(\tilde{x}_0, \ldots, \tilde{x}_n, \tilde{\lambda})\\
&=& 
\alpha \exp \left (c^{n-1}_nz - c^{2(n-1)}_ny - \alpha \sum^{n-1}_{l = 0}c^{n-1-l}_n\tilde{x}_l\right) - 
\tilde{\lambda}= H_n(\tilde{\lambda}).  
\end{eqnarray*}
Since $\hat{\lambda }^n$ is the unique solution of $H_n(\lambda) = 0$, 
we have that $\tilde{\lambda } = \hat{\lambda }^n$. 
This equality also implies 
$\tilde{x}_k = \mathcal {T}_k(\hat{\lambda }^n) = \hat{\psi}^n_k$, $k = 0, \ldots , n - 1$, 
and so the solution of (\ref {eg_temp_eqn1}) is also unique. 
\end{proof}

Now, we arrive at the following proposition. 

\begin{proposition} \ \label{prop_opt_hat_psi}It holds that 
$f^n = \tilde{f}^n(\hat{\psi }^n_0, \ldots , \hat{\psi }^n_{n-1}) / \alpha$ for sufficiently large values of $n$. 
\end{proposition}

\begin{proof} 
By Lemma \ref {lemma_divergence}, 
we can find an $M > 0$ large enough so that $\tilde{f}^n(x) < 0$ holds for $x\in \Xi ^n(\varphi )$ when $|x| \geq M$. 
Then, $\tilde{f}^n$ has at least one local maximum on $(-M, M)^n$, which we denote by $\tilde{x} = (\tilde{x}_0, \ldots , \tilde{x}_{n - 1})$. 
By the Lagrange multiplier method, 
we see that there is a $\tilde{\lambda}\in \mathbb {R}$ such that 
(\ref {eg_temp_eqn1}) holds at $(\tilde{x}, \tilde{\lambda})$. 
Then, Lemma \ref {lemma_Lagrange} implies 
$\tilde {x}_k = \hat{\psi}^n_k$ for $k = 0, \ldots, n - 1$. 
This means that $(\hat{\psi}_1, \ldots, \hat{\psi}_{n-1})$ is the unique local maximum, 
which also becomes the global maximum of $\tilde {f}^n$ on $\Xi ^n(\varphi)$. 
\end{proof}

Now, we prove Theorem \ref {eg_OU}. 
We write $\tilde{f}^n(\hat{\psi }^n_0, \ldots, \hat{\psi }^n_{n-1})$ as the following three parts: 
\begin{eqnarray*}
\tilde{f}^n(\hat{\psi }^n_0, \ldots , \hat{\psi }^n_{n-1}) &=& 
e^{z-y}(1-e^{-\alpha \hat{\psi }^n_0})\\&& + 
\sum^{n-2}_{k=1}
\exp \left(c^k_nz - c^{2k}_ny - \alpha \sum ^{k-1}_{l = 0}c^{k-l}_n\hat{\psi }^n_l\right)
(1-e^{-\alpha \hat{\psi}^n_k})\\&& + 
\exp \left (c^{n-1}_nz - c^{2(n-1)}_ny - \alpha \sum^{n-2}_{k = 0}c^{n-1-k}_n\hat{\psi }^n_k\right)
(1-e^{-\alpha \hat{\psi}^n_{n-1}})\\
&=:& 
\tilde{A}_n + \tilde{B}_n + \tilde{C}_n, 
\end{eqnarray*}
where $\tilde{A}_n$, $\tilde{B}_n$, and $\tilde{C}_n$ are defined in the obvious way. By Lemma \ref {conv_str}, we easily have 
\begin{eqnarray}\label{conv_final1}
\tilde{A}_n \longrightarrow e^{z-y}(1-e^{-\alpha p^*}), \ \ 
n\rightarrow \infty . 
\end{eqnarray} 
Using the relation (\ref {temp_lemma_eq_3}) and Lemmas \ref {lem_conv_F_inv}--\ref {conv_str}, we have 
\begin{eqnarray}\nonumber \label{conv_final2}
\tilde{C}_n &=& 
\exp \left ((c^{n-1}_n - c^{n-2}_n)z - c^{2(n-1)}_ny - 
\alpha F^{n, -1}_{n-2}(\exp (c^{2(n-1)}_ny)\hat{\lambda }^n/\alpha)\right )\\\nonumber 
&&\times (1-e^{-\alpha \hat{\psi}^n_{n-1}})\\\nonumber 
&\longrightarrow & 
\exp\left(e^{-2\beta}y - 
\alpha P^{-1}(\exp(e^{-2\beta }y)\lambda^*/\alpha)\right)
(1-e^{-\alpha q^*})\\
&=& 
e^{z-y}e^{-\alpha \eta^*_1}(1-e^{-\alpha q^*})
\end{eqnarray} 
as $n\rightarrow \infty$. 
To calculate the limit of $\tilde{B}_n$, we set 
\begin{eqnarray*}
\hat{B}_n = 
\frac{\alpha }{n}\sum^{n-2}_{k=1}
\exp \left (c^{2k}_ny - \alpha \xi ^*_{k/n}\right )\zeta ^*_{k/n}. 
\end{eqnarray*}
Then, we have 
\begin{eqnarray}\nonumber \label{conv_final3}
&&|\tilde{B}_n - \hat{B}_n| \leq 
e^z\left \{ \sum ^{n-2}_{k=1}
\left |e^{-\alpha \hat{\psi }^n_k} - e^{-\alpha \zeta ^*_{k/n}/n}\right | + 
\sum ^{n-2}_{k=1}
\left |1 - e^{-\alpha \zeta ^*_{k/n}/n} - \frac{\alpha \zeta ^*_{k/n}}{n}\right |\right \}\\\nonumber &&\hspace{5mm} + 
\frac{\alpha }{n}\sum ^{n-2}_{k=1}
\left| \exp(-c^{2k}_ny-\alpha F^{n, -1}_k(\exp (c^{2k}_ny)\hat{\lambda }^n/\alpha) - 
\exp (c^{2k}_ny-\alpha \xi ^*_{k/n})\right| \zeta^*_{k/n} \\
&&\ \ \ \ \ \longrightarrow  0, \ \ n\rightarrow \infty 
\end{eqnarray}
by virtue of (\ref {temp_lemma_eq_3}) and Lemmas \ref {lem_conv_F_inv}--\ref {conv_str}. 
Moreover, we have 
\begin{eqnarray}\nonumber \label{conv_final4}
\lim _{n\rightarrow \infty }\hat{B}_n = 
\alpha \int^1_0\exp (e^{-2\beta r}y - \alpha \xi^*_r)\zeta ^*_rdr = 
\alpha e^{z-y}\int^1_0e^{-\alpha \eta^*_r}\zeta^*_rdr. 
\end{eqnarray}
By (\ref {conv_final1})--(\ref {conv_final4}), 
we see that 
$w + e^{F+y}(\tilde{A}_n + \tilde{B}_n + \tilde{C}_n) / \alpha$ converges to 
the right-hand side of (\ref {rw_value_fnc}). 
Then, we obtain the desired assertion by Propositions \ref {prop_conv_fn} and \ref {prop_opt_hat_psi}. \qed

\end{document}